\documentclass[12pt]{article}

\usepackage{amsmath,amssymb}
\usepackage{graphicx}
\usepackage{fullpage}
\usepackage{hyperref}

\newtheorem{theorem}{Theorem}[section]
\newtheorem{lemma}[theorem]{Lemma}
\newtheorem{conjecture}{Conjecture}

\newcommand{\argmin}{\mathrm{argmin}}
\renewcommand{\d}{\textrm{d}}
\newcommand{\R}{\mathbb{R}}

\newcommand{\N}{\mathbb{N}}
\newcommand{\proof}{\noindent{\bf Proof:}~}
\newcommand{\qed}{\hfill{\bf QED}}

\newcommand{\E}{\mathcal{E}}
\newcommand{\V}{\mathcal{V}}

\newcommand{\cP}{\mathcal{P}}

\newcommand{\Prob}{\mathrm{Prob}}

\begin{document}

\title{Quantifying noisy attractors: from heteroclinic to excitable networks}

\author{Peter Ashwin\thanks{Center for Systems, Dynamics and Control/ Department of Mathematics, University of Exeter, Exeter EX4 4QF, UK and EPSRC Centre for Predictive Modelling in Healthcare, University of Exeter, Exeter, EX4 4QJ, UK. Email: P.Ashwin@exeter.ac.uk} and Claire Postlethwaite\thanks{ Department of Mathematics, University of Auckland, Auckland, 1142, New Zealand. Email: c.postlethwaite@auckland.ac.nz}}


\maketitle

\begin{abstract}
Attractors of dynamical systems may be networks in phase space that can be heteroclinic (where there are dynamical connections between simple invariant sets) or excitable (where a perturbation threshold needs to be crossed to a dynamical connection between ``nodes''). Such network attractors can display a high degree of sensitivity to noise both in terms of the regions of phase space visited and in terms of the sequence of transitions around the network. The two types of network are intimately related---one can directly bifurcate to the other.

In this paper we attempt to quantify the effect of additive noise on such network attractors. Noise increases the average rate at which the networks are explored, and can result in ``macroscopic'' random motion around the network. We perform an asymptotic analysis of local behaviour of an escape model near heteroclinic/excitable nodes in the limit of noise $\eta\rightarrow 0^+$ as a model for the mean residence time $T$ near equilibria. The heteroclinic network case has $T$  proportional to $-\ln\eta$ while the excitable network has $T$ given by a Kramers' law, proportional to $\exp (B/\eta^2)$. There is singular scaling behaviour (where $T$ is proportional to $1/\eta$) at the bifurcation between the two types of network.

We also explore transition probabilities between nodes of the network in the presence of anisotropic noise. For low levels of noise, numerical results suggest that a (heteroclinic or excitable) network can approximately realise any set of transition probabilities and any sufficiently large mean residence times at the given nodes. We show that this can be well modelled in our example network by multiple independent escape processes, where the direction of first escape determines the transition. This suggests that it is feasible to design noisy network attractors with arbitrary Markov transition probabilities and residence times.
\end{abstract}



\section{Introduction}
\label{sec:intro}

It is well known that noise can play a fundamental role in modifying the qualitative behaviour of a dynamical system. This is especially the case for what we term ``network attractors'' that include a number of invariant sets connected in some dynamical way. In this paper we consider the effect of noise on two related types of network attractor: heteroclinic networks (equilibria connected by heteroclinic orbits) and excitable networks (equilibria connected by orbits that start within some distance of the starting equilibrium). As noted in previous work~\cite{AshOroWorTow07,AshPos15}, a bifurcation of the equilibria in a symmetric system may lead to a transition from heteroclinic to excitable attractor.

For heteroclinic cycle attractors, it is well known that addition of noise can cause a non-ergodic attractor to become an approximately periodic ``noisy'' limit cycle \cite{StoHol90,StoArm99}. For excitable systems, the creative properties of noise in a potential landscape have been well studied in the literature on stochastic resonance \cite{Linder_etal_2004,Benzi81} where Kramers' law for escape times near a stable equilibrium coupled with global reconnection can lead to approximately periodic behaviour. Both of these effects can be thought of as a regularizing effect of adding noise. 

There is another noise-induced effect that seems to have received less attention (notable exceptions being \cite{ArmStoKir03,Bakhtinb}) If there is more than one outgoing direction for a connection (either heteroclinic, or excitable) from an equilibrium then it is not immediately clear which connection will be followed by the trajectory. On the one hand, there may be one preferred direction corresponding to the most unstable eigenvalue (in the case of a heteroclinic connection) or the shallowest potential saddle (in the case of an excitable connection). On the other hand, if the noise is anisotropic then variations in noise amplitudes in different directions can make one direction preferred over another. In fact, the connection chosen will be the result of a competition between the noise and dynamical processes for a number of possible outcomes. This results in a macroscopically observable randomness in the dynamics, where the noise forces dynamical behaviour to explore the network in a random manner.

The first aim of this paper is to present a qualitative exploration of the effect of additive noise on network attractors. We characterize the scaling of mean residence times near equilibria for both the heteroclinic and excitable cases, using a mixture of asymptotic analysis of a simplified problem and numerical examples. We also study how the noise determines the transition probabilities from a given node.

The second aim of this paper is to present a design principle for noisy networks. For a given (but arbitrary) set of mean residence times and transition probabilities, we argue it is possible to find a network attractor that is well-modelled by a first order Markov process where the mean residence times and transition probabilities are as desired. As previously shown \cite{AshPos13}, in the small noise limit, motion around a noisy network may be modelled as a one-step Markov chain as long as the local values of the eigenvalues do not cause ``lift-off'' and longer time correlations in the trajectory \cite{ArmStoKir03,Bakhtin,Bakhtinb}.

The paper is structured as follows: In Section~\ref{sec:3graph} we give two illustrative examples, where low amplitude noise added to a dynamical system that realises a ``network attractor'' gives rise to a random walk around the ``noisy network''. Section~\ref{sec:general} introduces network attractors for deterministic and noisy systems in general terms, along the lines of \cite{AshPos15}. We quantify the trajectory in terms of random variables for the residence times and transitions between network nodes. The means of these random variables give the mean residence time and the transition probabilities. In section~\ref{subsec:connecting} we give some general hypotheses on the nature of noisy network attractors and, assuming these hypotheses, we conjecture that any set of transition probabilities and sufficiently long mean residence times may be approximately realised by appropriate choice of noise amplitudes.

Section~\ref{sec:residence} models the mean residence times at each node by considering escape from a region near an equilibrium for the case where there is a connection in only one dimension. We find low-noise asymptotic scalings of the mean residence time on both sides of, as well as at the bifurcation between, heteroclinic and excitable connections. These scalings are verified and illustrated in Section~\ref{sec:onedimsims} using numerical simulations for a one dimensional SDE where there is transition from heteroclinic to excitable connection on changing a parameter.

Section~\ref{sec:switching} examines transition probabilities on a network. We consider a system with a simple (but fully nonlinear) noisy network attractor, adapting an example from our previous work~\cite{AshPos15}. For this example, we show that one can design the transition probabilities and mean residence times for a noisy network attractor by specifying the amplitudes of additive noise within the system. Details of the construction are included in Appendix~\ref{app:bidirthree}. 

Although the general problem of relating the transition probabilities to the noise amplitudes seems to be difficult, it seems that the switching can be well-modelled as a competition between two independent escape processes, and we investigate this in Section~\ref{sec:multiplesc}. In the case where the escape distribution is close to exponential we show that one can approximate the transition probability simply from the mean escape times. More generally the transition probability is determined by the distributions of escape times, not just their means.

Finally Section~\ref{sec:discuss} gives a discussion of some implications, possible areas of application, and open questions raised by this work. 

\subsection{Example: random walks on three-node networks}
\label{sec:3graph}

In order to motivate the sort of dynamics we are considering, consider the finite graphs shown in Figure~\ref{fig:3graphs}. Appendix~\ref{app:3graphs} describes dynamical systems of the form described in our previous work~\cite{AshPos15} that realise each of the networks shown in Figure~\ref{fig:3graphs}; see equations~\eqref{eq:C3system1} and~\eqref{eq:C3system2} for details. The aim of this paper will be to quantify both the mean residence times and the transition probabilities for such a network attractor in the presence of noise.

\begin{figure}%
\begin{center}
\setlength{\unitlength}{1mm}
\begin{picture}(110,50)(0,0)
\put(5,5){\includegraphics[trim= 0cm 0cm 0cm 0cm,clip=true,width=50mm]{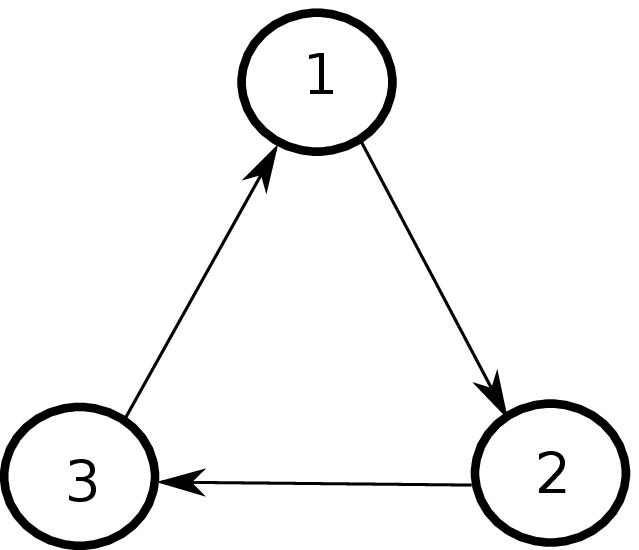}}
\put(65,5){\includegraphics[trim= 0cm 0cm 0cm 0cm,clip=true,width=50mm]{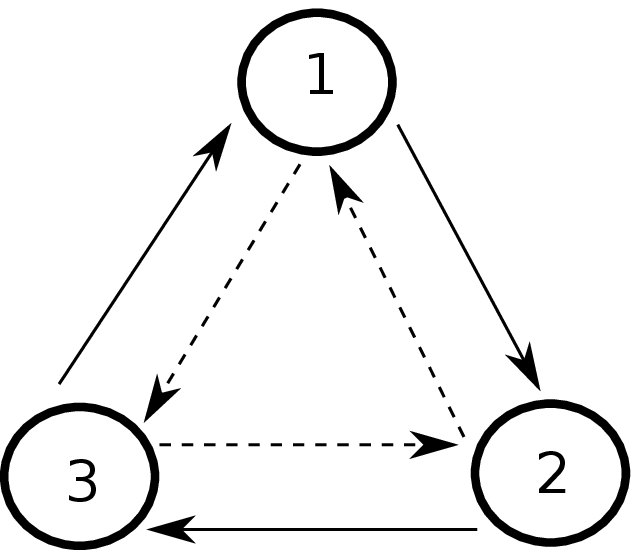}}
\put(5,45){(a)}
\put(65,45){(b)}
\end{picture}
\end{center}
\caption{Two graphs with three nodes. In (a) transitions can only be made in one direction, and in (b), transitions may be made in both directions. We realise each of these graphs as heteroclinic/excitable network attractors with added noise using the system described in Appendix A. In both cases the mean residence time at equilibria is infinite unless there is noise. For case (a) we add noise of strength $\eta_{cw}$ in the direction of the connections. In case (b) we add noise of strengths $\eta_{cw}$ and $\eta_{acw}$ to transitions in the clockwise and anticlockwise directions, respectively. See Figure~\ref{fig:runs} for example timeseries showing trajectories near noisy network attractors that realise these graphs.}%
\label{fig:3graphs}%
\end{figure}

We show in Figure~\ref{fig:runs} some numerical simulations of typical runs starting at node $\xi_1$. The one dimensional observable $S(t)$ (see Appendix~\ref{app:3graphs}) has the property that $S(t)\approx k$ whenever the trajectory is near the equilibrium $\xi_k$. The components $p_j$ are approximately equal to $1$ when the trajectory is near the equilibrium $\xi_k$, and the components $y_j$ are non-zero during the transitions between equilibria. Here $p_1$ corresponds to the transition from $\xi_1$ to $\xi_2$.
Figures~\ref{fig:runs}(a) and (c) show heteroclinic and excitable realisations, respectively, for the uni-directional cycle shown in Figure~\ref{fig:3graphs}(a). 
Figures~\ref{fig:runs}(b) and (d) show heteroclinic and excitable realisations, respectively, for the bi-directional cycle shown in Figure~\ref{fig:3graphs}(b). Here, transitions are possible in both directions. 
Close inspection reveals that there is much greater variability in residence times for the excitable realisations than for the heteroclinic realisation. In particular, the time-series for the heteroclinic uni-directional ring (Figure~\ref{fig:runs}(a)) is approximately periodic.
In the noise-free case (not shown), excitable realisations remains at the (now stable) starting state, while heteroclinic realisations perform an asymptotic slowing down between the nodes in the graph.

\begin{figure}%
\begin{center}
\includegraphics[width=17cm,trim={0cm 0cm 0cm 0cm},clip=]{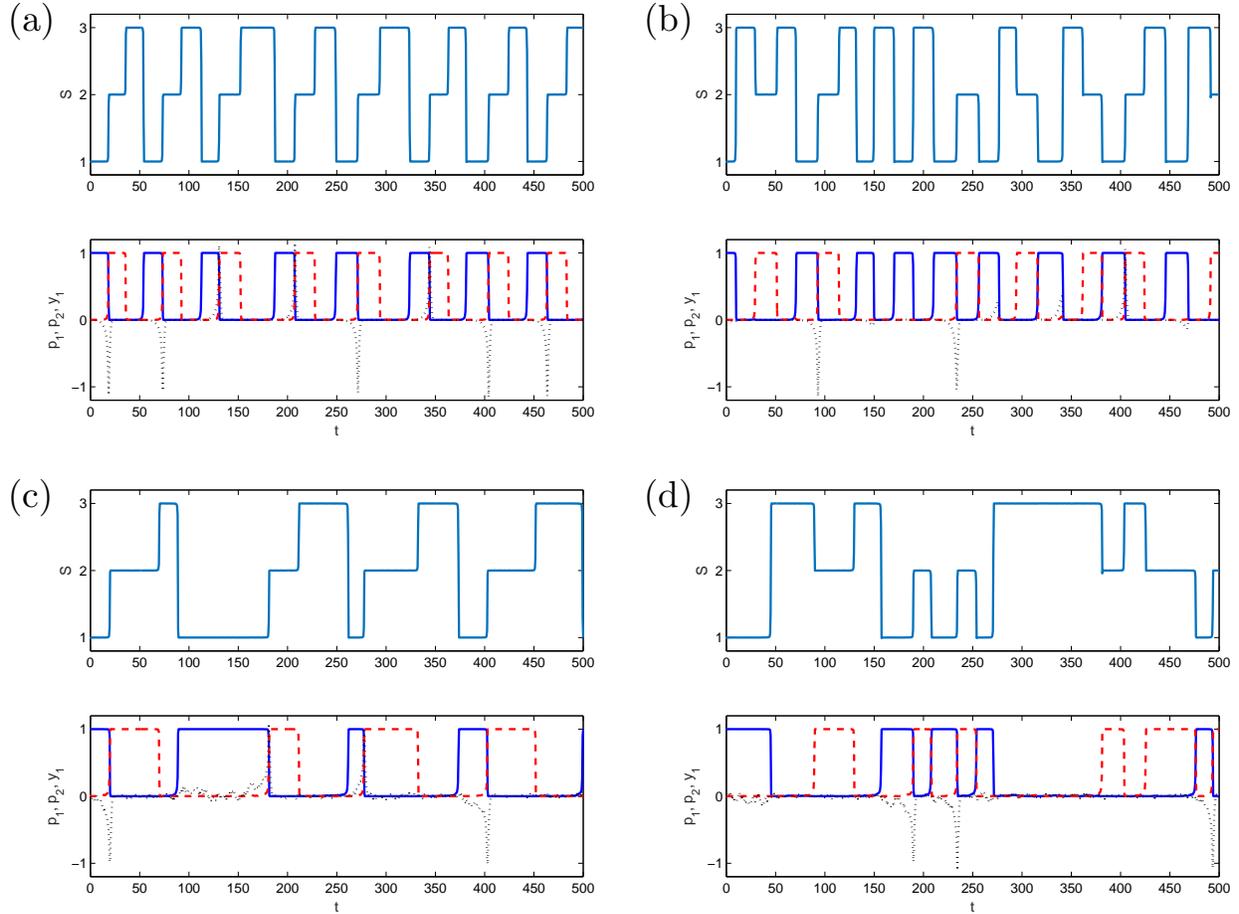}~
\end{center}
\caption{Timeseries of $S(t)$ (top), and $p_1$, $p_2$ and $y_1$ (bottom) for typical trajectories corresponding to realisations of the graphs shown in Figure~\ref{fig:3graphs} with weak noise; the system and parameters are described in Appendix~\ref{app:3graphs}. If the system state is close to the equilibrium $\xi_k$ that represents the $k$th node in the graph then the observable satisfies $S(t)\approx k$. For these runs we show time series for $\eta_{cw}=\eta_{acw}=0.03$.  Panels (a) and (c) are for the uni-directional ring (shown in Figure~\ref{fig:3graphs}(a)), and panels (b) and (d) are for the bi-directional ring  (shown in Figure~\ref{fig:3graphs}(b)). Panels (a) and (b) show heteroclinic realisations and panels (c) and (d) show excitable realisations of the graphs in Figure~\ref{fig:3graphs}.}%
\label{fig:runs}%
\end{figure}

\section{Deterministic and noisy network dynamics}
\label{sec:general}

Consider an autonomous ordinary differential equation (ODE)
\begin{equation}
\frac{d}{dt} x= f(x,\nu)
\label{eq:ode}
\end{equation}
on $x\in\R^d$ where $t\geq 0$, $f(x,\nu)$ is a smooth nonlinear function, and $\nu\in\R$ is a bifurcation parameter. We first define more precisely what we mean by heteroclinic and excitable networks in such a deterministic system before considering the statistics of noise-perturbed versions.

\subsection{Networks in phase space}

We say there is a {\em heteroclinic connection} from one equilibrium $\xi_i$ to another $\xi_j$  for (\ref{eq:ode}) if
$$
W^u(\xi_i)\cap W^s(\xi_j)\neq \emptyset.
$$
We say (\ref{eq:ode}) has a {\em heteroclinic network attractor} if there is an asymptotically stable compact connected set $\Sigma\subset\R^d$ such that for some set of saddle equilibria $\{\xi_k\}_{k=1}^{N}$ we have
\begin{equation}
\Sigma= \bigcup_{k=1}^{N} W^u(\xi_k)
\label{eq:hetnet}
\end{equation}
where
$$
W^u(\xi)=\{x~:~\alpha(x)=\{\xi\}\},~~W^s(\xi)=\{x~:~\omega(x)=\{\xi\}\}
$$
(these sets are manifolds if the saddles are hyperbolic). This definition is fairly weak (cf \cite{AshPos15}) - we do not necessarily assume hyperbolicity of the saddles or even chain recurrence of the network. However we assume that the closure of all $W^u(\xi_k)$ are contained within the network (the network is "clean" \cite{field_96}) as we will be concerned with behaviour that remains close to the network under stochastic perturbation.

We say the system (\ref{eq:ode}) has an {\em excitable connection for amplitude $\delta>0$} from  one equilibrium $\xi_i$ to another $\xi_j$ if
$$
B_{\delta}(\xi_i)\cap W^s(\xi_j)\neq \emptyset.
$$
This connection has {\em threshold} $\delta_{th}$ if
$$
\delta_{th}= \inf \{\delta>0~:~B_{\delta}(\xi_i)\cap W^s(\xi_j)\neq \emptyset\}.
$$
A set $\Sigma$ is an {\em excitable network for amplitude $\delta>0$} \cite{AshPos15} if there is a set of equilibria $\{\xi_i\}$ such that 
\begin{equation}
\Sigma=\Sigma_{\mathrm{exc}}(\{\xi_i\},\delta):=\bigcup_{i,j=1}^{n} \{\phi_t(x)~:~x\in B_{\delta}(\xi_i) \mbox{ and }t>0\}\cap W^s(\xi_j)
\label{eq:exnet}
\end{equation}
As noted in \cite{AshPos15}, a heteroclinic connection is also an excitable connection with $\delta_{th}=0$ though the converse is not the necessarily true.

An excitable network for amplitude $\delta$ means if we can follow an arbitrary path on the network by a mixture of trajectories and ``jumps'' of maximum size $\delta$. In a previous paper \cite{AshPos15} we gave a particular construction of coupled nonlinear systems (\ref{eq:ode}) where an arbitrary network can be constructed as a heteroclinic or as an excitable network in phase space.

\subsection{Noisy network attractors}

For cases where the noise-free system (\ref{eq:ode}) has a network attractor $\Sigma$ we will investigate the associated autonomous stochastic differential equation (SDE)
\begin{equation}
dx= f(x,\nu)dt+\eta dw
\label{eq:sde}
\end{equation}
on $x\in\R^d$ where $t\geq 0$, $w(t)$ is $d$-dimensional Brownian motion, and $\eta=\mbox{diag}(\eta_1,\ldots,\eta_d)$. We are concerned with investigating the influence of noise on the associated noisy system (\ref{eq:sde}) under the assumption that trajectories remain close to the heteroclinic or excitable network attractor. We express this more precisely in Section~\ref{subsec:connecting}.

In what follows we will consider $\vartheta\in \Omega$ where $\Omega$ represents the possible noise trajectories. Formally one can understand the solution of (\ref{eq:sde}) as a ``random dynamical system'' \cite{Arnold98}: the solution $x(t)$ can be viewed as a cocycle over the noise trajectory: we write
\begin{equation}
\begin{array}{rcl}
x(t+s) &=& \varphi(t,\vartheta(s),x(s))\\
\vartheta(t+s) &=& \theta(t,\vartheta(s))
\end{array}
\label{eq:rds}
\end{equation}
for any $t,s\geq 0$: note that $\varphi:\R^+\times\Omega\times \R^d\rightarrow \R^d$ is a cocycle that represents the evolution of the system with noise $\vartheta(s)$ whilst $\theta:\R^+\times \Omega\rightarrow \Omega$ represents the evolution of the noise - typically just a shift in time. We will assume there is a measure $\mu_\Omega$ (such as Wiener measure) on $\Omega$, and we assume $\vartheta$ is chosen from a set of full measure with respect to $\mu_\Omega$. We will also assume that the random dynamical system has an attractor that supports a natural ergodic measure $M$ on $\Omega\times \R^d$ whose projection onto $\Omega$ is $\mu_\Omega$ and whose marginals are absolutely continuous with respect to $d$-dimensional Lebesgue measure on the fibres $\R^d$. For any $A\subset \R^d\times \Omega$ we write $\Prob(A):=M(A)$. In heuristic terms we can think of $\Prob(\cdot)$ as assigning probabilities to possible asymptotic states of the noisy system.

\subsection{Itineraries on attracting networks}
\label{sec:itineraries}

Let us assume that typical trajectories of (\ref{eq:sde}) spend most of their time close to a network $\Sigma$ of the form (\ref{eq:hetnet}) or (\ref{eq:exnet}). We attempt to describe the motion in terms of the itinerary around the network, i.e. the sequence and timing of visits to the equilibrium nodes $\xi_k$.

Fix a tolerance $h>0$ (such that $|\xi_p-\xi_q|>2h$ for all $p\neq q$) and define
$$
K(x) := \left\{\begin{array}{rl}
i & \mbox{ if } |x-\xi_i|\leq h\\
0 & \mbox{ otherwise.}\end{array}\right.
$$
For a trajectory $x(t)$ we define
$$
\tilde{K}(t) = \{ K(s) ~:~ s = \sup\{ s\leq t~:~ K(s)\neq 0\} \}
$$
which gives the ``last visited node'' and if we start near a node this will always be non-zero. if $\tilde{K}(x(t))=i$ we say $x(t)$ is {\em close to the $i$th node}.

For $|\eta|$ small and trajectories that remain close to $\Sigma$ we expect that
$$
\lim_{T\rightarrow \infty} \frac{1}{T} \int_{s=0}^T |K(x(s)) -\tilde{K}(s)|\,ds
$$
to be small, i.e. $K(x(t))=\tilde{K}(t)$ most of the time. For a given initial condition $x_0$, amplitude and realisation  of the noise $\vartheta$, the trajectory $x(t)$ divides up the $t>0$ into an {\em itinerary}. This is the unique sequence of {\em epochs}
$$
\{ (i_j(x_0,\vartheta),s_j(x_0,\vartheta)) ~:~ j\in\N\}
$$
such that $\tilde{K}(t)=i_j$ for the interval $t\in[s_j,s_{j+1})$, and $i_{j+1}\neq i_j$. As in \cite{AshPos13}, the times of entry $s_j$ are increasing while the {\em duration} of the $j$th epoch we define to be 
$$
\tau_j(x_0,\vartheta)=s_{j+1}(x_0,\vartheta)-s_j(x_0,\vartheta).
$$

We are interested in various statistics of this itinerary including the distribution of {\em residence times} for the $j$th node:
$$
\rho_j(\tau)=\Prob\left(\{(\tilde{x},\tilde{\vartheta})~:~\tau_\ell=\tau \mbox{ given that }i_\ell(\tilde{x},\tilde{\vartheta})=j\}\right).
$$
The {\em mean residence time} at the $j$th node is the expected value of $\tau_j$, i.e.
\begin{equation}
T_j = \int_{\tau=0}^{z\infty} \tau \rho_j(\tau)\,d\tau
\end{equation}
If the network has several outgoing connections from a node one might expect the addition of noise to enhance random switching; we show that this can, at least in our case, be well modelled as a competition between independent first escape time processes in the different directions, such that the residence time is the minimum first escape time and the transition probability is the direction of first escape.

For a given (finite) sequence of nodes $\{j_k~:~k=1,\cdots m\}$ we can examine the probability of seeing this sequence of nodes as
\begin{equation}\label{eq:prob}
\cP(j_1,\ldots,j_m) := \Prob \left( \{(\tilde{x},\tilde{\vartheta})~:~i_{\ell}(\tilde{x},\tilde{\vartheta})=j_{\ell} \mbox{ for }\ell=1,\ldots,m\} \right)
\end{equation}
and use this to investigate the asymptotic probabilities being at state $j$, 
$\pi_j:=\cP(j)$ (assuming that $\pi_j>0$). In equation~\eqref{eq:prob}, we can think of the the initial condition of the trajectory $\tilde{x}$ being chosen randomly from the attractor, and then the probability is taken with respect to that initial condition and all possible noise trajectories.

More precisely, the {\em transition probability} that the next state is $j_2$ given we are at state $j_1$ is
\begin{equation}
\pi_{j_1,j_2}:=\frac{1}{\pi_{j_1}}\cP(j_1,j_2)
\end{equation}

As in \cite{AshPos13} we say the {\em transitions are memoryless} if
\begin{equation}
\cP(j_1,\ldots,j_m)=\cP(j_1,\ldots,j_{m-1}) \pi_{p,q}
\label{eq:mless}
\end{equation}
for all $p,q$ and any sequence $j_1,\cdots,j_m$ where $j_{m-1}=p$ and $j_m=q$. As noted in \cite{AshPos13}, in many cases we can expect the transitions to be asymptotically memoryless (i.e. (\ref{eq:mless}) holds with an error that goes to zero as the noise goes to zero), in which case the transitions are well modelled by a first order Markov chain where the transition probabilities are $\pi_{p,q}$.

More precisely we say for some $\epsilon>0$ that the transitions are {\em $\epsilon$-memoryless} if
\begin{equation}
|\cP(j_1,\ldots,j_m)-\cP(j_1,\ldots,j_{m-1}) \pi_{p,q}|<\epsilon
\label{eq:deltamless}
\end{equation}
for all $p,q$ and any sequence $j_1,\cdots,j_m$ where $j_{m-1}=p$ and $j_m=q$.

\subsection{Connecting microscopic and macroscopic randomness}
\label{subsec:connecting}

An important question that we aim to address in the remainder of this paper is to understand how random variables that determine the itineraries of trajectories of a noisy network attractor for (\ref{eq:sde}) are influenced by the dynamics of the noise-free system (\ref{eq:ode}) and the noise amplitudes. In particular, we are concerned with systems where in the limit of asymptotically low additive noise, all of the mass of the attractor is centred on the network nodes. More precisely, we consider systems of the form (\ref{eq:sde}) such that 
\begin{itemize}
\item[(H1)] The noise-free system has a network attractor $\Sigma$ between a finite set of equilibria $\{\xi_i\}_{i=1}^N$ and any connection from $\xi_i$ to $\xi_j$ has added noise of amplitude $\eta_{ij}$
\item[(H2)] For fixed tolerance $h>0$ and any $\epsilon>0$, there is an $\eta>0$ such that whenever $|\eta_{ij}|<\eta$ for all $i,j$ any typical trajectory $x(t)$ with itinerary $K(t)$ will satisfy
$$
\frac{1}{t} \int_{s=0}^{t} \delta_{\tilde{K}(x(s)),K(x(s))} \, ds <\epsilon
$$
i.e. the proportion of time where trajectory is not close to an equilibrium is arbitrarily small.
\item[(H3)] For any $\epsilon>0$ there is an $E$ such that whenever $\eta_{ij}<E$ for all $i,j$ then the transitions are $\epsilon$-memoryless.
\end{itemize}

Making the above assumptions, we conjecture that the (microscopic) noise amplitudes $\eta_{ij}>0$ can be chosen to realise (macroscopic) noisy network dynamics with any given statistics (that is, mean residence times $T_j$ and transition probilities $\pi_{i,j}$), as long as the residence times are sufficiently long. We believe that (H1)-(H3) are reasonable assumptions to make, and in particular, can be numerically verified for the example networks we give in Section~\ref{sec:3graph}. Bakhtin has results \cite[Theorem 6.1]{Bakhtinb} for the limiting invariant measure for some heteroclinic cycles, that implies (H2). Hypothesis (H3) is discussed in more detail in our previous work~\cite{AshPos13} and also by Bakhtin~\cite[Section 10]{Bakhtinb}. (H3) can be violated for heteroclinic networks, if parameters are chosen so that there is `lift-off'~\cite{ArmStoKir03}. This may be the case if there are outgoing eigenvalues that are stronger than the incoming eigenvalues at an equilibrium. For excitable networks we do not expect (H3) to be easily violated.

\begin{conjecture}
\label{conj:main}
Suppose that (\ref{eq:sde}) has a noisy network attractor such that hypotheses (H1)-(H3) hold. We conjecture there is a $\tau>0$ such that for any desired mean residence times $R_j>\tau$ and any desired transition probabilities $\Pi_{ij}>0$ with $\sum_j\Pi_{ij}=1$, there exists a choice of noise amplitudes $\eta_{ij}$ such that 
$$
T_j=R_j,~~\pi_{i,j}=\Pi_{i,j}
$$
for all $i,j$.
\end{conjecture}

We present some snumerical evidence supporting this in Section~\ref{sec:switching}.

\section{Residence times for noisy heteroclinic and excitable networks}
\label{sec:residence}

For the noisy network dynamics discussed in Section~\ref{sec:itineraries} we study the behaviour near the connections in terms of an escape process near an equilibrium. On entering a neighbourhood of $\xi_1$ the dynamics of those $y_i$ variables that correspond to outgoing directions in the graph will be unstable (for the heteroclinic case) or marginally stable (for the excitable case); we assume all others are strongly stable. Without loss of generality we consider $y=y_1$ corresponding to a connection from $\xi_1$ to $\xi_2$. The mean escape time from a neighbourhood of an equilibrium $\xi_1$ of a network will be approximated using a one dimension model of the bifurcation to an excitable connection. 

For the excitable case this is the well-studied Kramers escape rate from a local potential well. Although Kramers' result has been known and applied in many areas for a long time, only recently have full mathematical justifications of the asymptotic formulae been available \cite{Friedlinetal2012}, and generalisations to more complex situations including some bifurcation problems have only recently been developed by Bakhtin \cite{Bakhtin}, Berglund, Gentz \cite{Berglund2013,BerglundGentz2008} and others. A related case of escape over a potential maximum that undergoes a supercritical pitchfork bifurcation is analysed in detail by Berglund and Gentz~\cite{BerglundGentz2008}: we treat however the problem of escape from a saddle that becomes a sink at a subcritical pitchfork bifurcation on varying $\nu$. 

To this end, consider the one dimensional SDE
\begin{equation}
dx= -V'(x) dt + \eta dw.
\label{eq:sdeescape}
\end{equation}
Kramers' formula is an asymptotic formula for the mean transition time from one minimum, $x_0$, to another minimum, $y_0$, of $V(x)$ that causes the trajectory to pass over the maximum potential barrier $z_0$. It states that
\begin{equation}
T \approx \frac{2\pi}{\sqrt{V''(x_0)|V''(z_0)|}} \exp \left(2\frac{V(z_0)-V(x_0)}{\eta^2}\right)\left(1+O(\eta)\right)
\label{eq:Kramers}
\end{equation}
in the limit $\eta\rightarrow 0^+$ (see \cite{Berglund2013} for a review).

More precisely, we approximate the mean residence time near a saddle as the mean escape time $T(\nu,\eta)$ from $x=0$ for the one-dimensional problem (\ref{eq:sdeescape}) with potential 
\begin{equation}
V(x)=\frac{1}{6} x^6 -\frac{1}{2} x^4+\frac{\nu}{2}x^2
\label{eq:potential}
\end{equation}
from the interval $[-a,a]$ for some fixed $a$ of order one; more precisely we choose an $0<a$ that separates the additional potential wells of (\ref{eq:potential}) from $x=0$. Figure~\ref{fig:potls} illustrates the potential and the choice of $a$: we will be interested in cases where $\nu$ is close to zero so any additional equilibria lie within $[-a,a]$ and the noise amplitude is asymptotically small: $\eta\rightarrow 0^+$.

\begin{figure}%
\centerline{
\includegraphics[width=0.8\columnwidth]{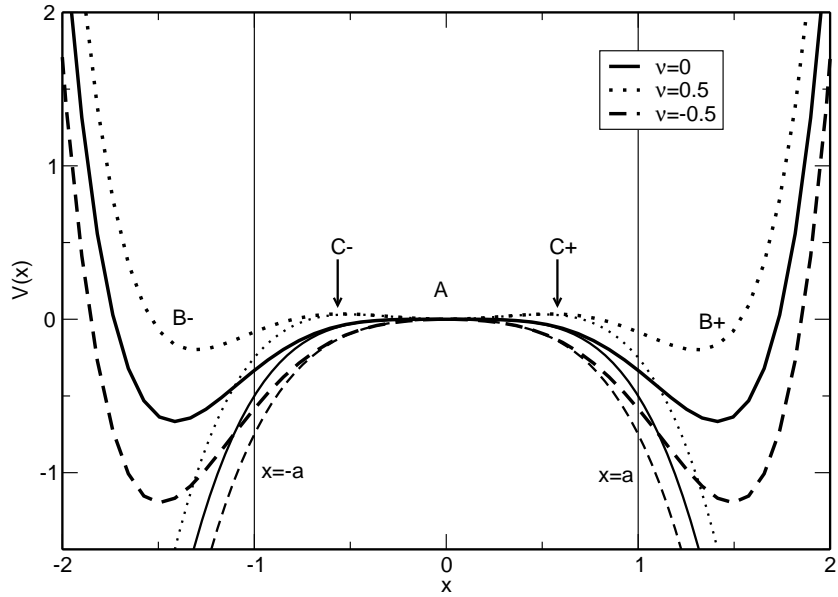}
}%
\caption{Bold lines show the potential $V(x)$ versus $x$ (\ref{eq:potential}) for values of $\nu<0$, $\nu=0$ and $\nu>0$. Note that there are minima at $B+$ and $B-$ for all $\nu$ close enough to zero, while $A$ is a local maximum for $\nu\leq 0$. For $\nu>0$ there are local maxima at $C+$ and $C-$ and $A$ is a local minimum. The fainter lines show the truncated potential (\ref{eq:mpotential}). We will model the transitions from $A$ to $B\pm$ in the full potential as the first passage through the lines $x=\pm a$ for the truncated potential; this gives good predictions for small enough $\sqrt{\nu}\ll a$ in the limit $\eta\rightarrow 0^+$.}%
\label{fig:potls}%
\end{figure}

We can consider the modified potential 
\begin{equation}
V(x)= \frac{\nu x^2-x^4}{2}.
\label{eq:mpotential}
\end{equation}
For $\nu<0$, this has a saddle at $x=0$ that is stabilised via a subcritical pitchfork on increasing $\nu$ through zero. For the case $\nu>0$, $V$ has a minimum at $x=0$ and maxima at $x_{\pm}=\pm \sqrt{\nu/2}$, and we assume $x_{\pm}\in[-a,a]$. Using Berglund~\cite[Section 3.1]{Berglund2013} we calculate the mean escape time, $w_{a}(x)$, for solutions of (\ref{eq:sdeescape}) starting at a location $x\in (-a,a)$ out of the interval $[-a,a]$. This is given by solving the Poisson problem
\begin{equation}
\frac{\eta^2}{2} \frac{d^2}{dx^2}w_a(x) - V'(x) \frac{d}{dx} w_a(x)=-1,~~ w_a(-a)=w_a(a)=0.
\label{eq:waitingtime}
\end{equation}
The solution of this can be expressed in integral form as
$$
w_a(s) = \frac{2}{\eta^2} \int_{x=s}^{a} \int_{y=0}^{x} \exp \frac{2(V(x)-V(y))}{\eta^2} \,dy\,dx
$$
For escapes with $x$ near the origin, i.e. $0<|x|\ll |a|$ of the potential (\ref{eq:potential}) this can be approximated by
\begin{equation}
T(\nu,\eta) =  \frac{2}{\eta^2} \int_{x=0}^{a} \int_{y=0}^{x} \exp \frac{\nu(x^2-y^2)+(y^4-x^4)}{\eta^2} \,dy\,dx.
\label{eq:escapes}
\end{equation}
In the following sections, we compute asymptotics of $T(\nu,\eta)$ for small $\eta$ and $\nu$. In particular, we consider the limit $\eta\rightarrow 0^+$ for three cases: $\nu<0$, $\nu=0$, and $\nu>0$. We begin by finding some bounds on $T(\nu,\eta)$.

\begin{lemma}
\begin{equation} \label{eq:Tbounds} 
 \frac{1}{2\eta} \int_{z=0}^{\frac{a^2}{\eta}} \frac{1-\exp(z(\alpha-z))}{z(z-\alpha)} \,dz < T(\nu,\eta) < \frac{1}{\eta} \int_{z=0}^{\frac{2a^2}{\eta}} \frac{1-\exp(2z(\alpha-z))}{2z(z-\alpha)} \,dz
\end{equation}
\end{lemma}

\begin{proof}
Rescale $v:=y/\sqrt{\eta}$, $u=x/\sqrt{\eta}$, 
and define
$
\alpha:=\frac{\nu}{\eta}
$
to get
\begin{equation}
T(\nu,\eta) =  \frac{2}{\eta} \int_{u=0}^{\frac{a}{\sqrt{\eta}}} \int_{v=0}^{u} \exp \left[\alpha (u^2-v^2)+(v^4-u^4)\right] \,dv\,du.
\end{equation}
Now let us define
$$
p:=u+v,~~q:=u-v.
$$
Changing integration variables from $(u,v)$ to $(p,q)$ we have 
\begin{equation}
T(\nu,\eta) =  \frac{1}{\eta} \int_{p=0}^{\frac{2a}{\sqrt{\eta}}} \int_{q=0}^{\min(p,\frac{2a}{\sqrt{\eta}}-p)} \exp \left[pq(\alpha -(p^2+q^2)/2)\right] \,dq\,dp.
\label{eq:ovTpq}
\end{equation}
In the region of integration we have $0<q<p$, so  $p^2<p^2+q^2<2p^2$ and so
$$
pq(\alpha-p^2)<pq(\alpha -(p^2+q^2))<pq(\alpha- p^2/2)
$$
We can thus find an upper bound to (\ref{eq:ovTpq}) using
\begin{eqnarray*}
T(\nu,\eta) & < & 
\frac{1}{\eta} \int_{p=0}^{\frac{2a}{\sqrt{\eta}}} \int_{q=0}^{\min(p,\frac{2a}{\sqrt{\eta}}-p)} \exp \left[q p (\alpha  - p^2/2)\right] \,dq\,dp\\
&<& \frac{1}{\eta} \int_{p=0}^{\frac{2a}{\sqrt{\eta}}} \int_{q=0}^{p} \exp \left[q p (\alpha  - p^2/2)\right] \,dq\,dp\\
&=&
\frac{1}{\eta} \int_{p=0}^{\frac{2a}{\sqrt{\eta}}} \frac{1-\exp [p^2(\alpha-p^2/2)]}{p(p^2/2-\alpha)} \,dp
\end{eqnarray*}
Changing coordinates to $z=p^2/2$, gives the upper bound.

A lower bound to (\ref{eq:ovTpq}) is given by
\begin{eqnarray*}
T(\nu,\eta) & > & 
\frac{1}{\eta} \int_{p=0}^{\frac{a}{\sqrt{\eta}}} \int_{q=0}^{\min(p,\frac{2a}{\sqrt{\eta}}-p)} \exp \left[q p (\alpha  - p^2)\right] \,dq\,dp\\
&>& \frac{1}{\eta} \int_{p=0}^{\frac{a}{\sqrt{\eta}}} \int_{q=0}^{p} \exp \left[q p (\alpha  - p^2)\right] \,dq\,dp\\
&=&
\frac{1}{\eta} \int_{p=0}^{\frac{a}{\sqrt{\eta}}} \frac{1-\exp [p^2(\alpha-p^2)]}{p(p^2-\alpha)} \,dp
\end{eqnarray*}

Changing coordinates to $z=p^2$ gives the lower bound. \hfill$\square$
\end{proof}

\subsection{Scaling for heteroclinic connections}

Heteroclinic connections in a network correspond to $\nu<0$. In this parameter regime the scaling for $T(\nu,\eta)$ as $\eta$ tends to zero is given as follows:
\begin{lemma}
Suppose $\nu<0$. Pick some $0<\beta<1$. Then in the limit $\eta\rightarrow 0^+$, 
\begin{equation}
\beta<\frac{T(\nu,\eta)}{\frac{1}{\nu}\ln \eta} <1.
\label{eq:nultzero}
\end{equation}
\end{lemma}
Observe that the leading order of this scaling is as expected from Stone and Holmes~\cite{StoHol90}.

~

\begin{proof}
We begin by computing the upper bound. 
For $\nu<0$ and $\eta>0$ (so that $\alpha=\nu/\eta<0$) note that the integrand in the upper bound in (\ref{eq:Tbounds}), $f(z)=\frac{1-\exp(2z(\alpha-z))}{2z(z-\alpha)}$, satisfies both
\[ f(z)<\frac{1}{2z(z-\alpha)}\quad \forall z>0 \]
and
\[
f(z)<1\quad  \forall z>0.
\]
This implies that for some $z^*>0$, we can split the integral into
\begin{eqnarray}
T(\nu,\eta) &<& \frac{1}{\eta} \left[ \int_{z=0}^{z^*} \,dz + \int_{z^*}^\frac{2a^2}{\eta} \frac{1}{2z(z-\alpha)} \,dz\right] \nonumber \\
&=&  \frac{z^{*}}{\eta} +\frac{1}{2\nu}\left[-\ln \frac{|\alpha|}{z^*}   -\ln \left(1+\frac{z^*}{|\alpha|}\right) + \ln \left(1-\frac{\nu}{2a^2}\right) \right]. \label{eq:upperbd2}
\end{eqnarray}
We then choose $z^*=\dfrac{1}{|\alpha|}=\dfrac{\eta}{|\nu|}=-\dfrac{\eta}{\nu}$ and letting $\eta\rightarrow 0^+$, we find
\begin{eqnarray*}
T(\nu,\eta) & < & -\frac{1}{\nu} -\frac{1}{2\nu}\ln \frac{|\nu|^2}{\eta^2}-\frac{1}{2\nu}\ln\left(1+\frac{\eta^2}{|\nu|^2} \right)+\frac{1}{2\nu}\ln\left(1-\frac{\nu}{2a^2} \right) \\
& < & \frac{1}{\nu}\ln\eta +\frac{1}{\nu}\left(-1-\ln|\nu|+\ln\left(  1-\frac{\nu}{2a^2}\right) \right) +O\left(\frac{\eta^2}{\nu^3} \right)
\end{eqnarray*}
so that in the limit $\eta\rightarrow 0^+$ for fixed $\nu<0$ and $a>0$,
\begin{equation}
T(\nu,\eta)< \frac{1}{\nu}\ln \eta +  K_1 + O(\eta^2)
\label{eq:nultzeroupper}
\end{equation}
 where
$$
K_1= \frac{1}{\nu}\left(-1-\ln|\nu|+\ln\left(  1-\frac{\nu}{2a^2}\right) \right). 
$$


We now obtain a lower bound. Let the integrand in the lower bound in (\ref{eq:Tbounds}) be $g(z)=\frac{1-\exp(z(\alpha-z))}{z(z-\alpha)}$ and fix some $0<\beta<1$. It can be shown that for $z^*(\beta)=-\frac{\ln(1-\beta)}{|\alpha|}$, the integrand satisfies
\[
g(z)>\frac{\beta}{z(z-\alpha)}\quad \mathrm{for}\ z>z^*(\beta)
\]
and
\[
g(z)>\frac{\beta}{-2\ln(1-\beta)} \quad \mathrm{for}\ 0<z<z^*(\beta)\quad \mathrm{and}\ |\alpha|>z^*(\beta)
\]

We can thus, for fixed $\beta$ and large enough $|\alpha|$, split the integral into
\begin{eqnarray*}
T(\nu,\eta) &>& \frac{1}{2\eta} \left[ \int_{z=0}^{z^*} \frac{\beta}{-2\ln(1-\beta)} \,dz +\beta \int_{z^*}^\frac{a^2}{\eta} \frac{1}{z(z-\alpha)} \,dz.\right]\\
&=&  \frac{\beta}{-4\ln (1-\beta) }\frac{z^{*}}{\eta} +\frac{\beta}{2\nu}\left[ -\ln \frac{|\alpha|}{z^*}   -\ln \left(1+\frac{z^*}{|\alpha|}\right) + \ln \left(1-\frac{\nu}{a^2}\right) \right].
\end{eqnarray*}
Then, substituting for $z^*(\beta)=-\frac{\ln(1-\beta)}{|\alpha|}=-\ln(1-\beta) \frac{\eta}{|\nu|}$, we find
\begin{eqnarray*}
T(\nu,\eta) & > & -\frac{\beta}{4\nu} -\frac{\beta}{2\nu}\ln\left( \frac{|\nu|^2}{-\ln(1-\beta)\eta^2}\right)-\frac{\beta}{2\nu}\ln\left(1-\ln(1-\beta)\frac{\eta^2}{|\nu|^2} \right)+\frac{\beta}{2\nu}\ln\left(1-\frac{\nu}{a^2} \right) \\
& = & \frac{\beta}{\nu}\ln\eta +\frac{\beta}{\nu}\left(-\frac{1}{4}-\ln|\nu|+\ln(-\ln(1-\beta))+\frac{1}{2}\ln\left(  1-\frac{\nu}{a^2}\right) \right) +O\left(\frac{\eta^2}{\nu^3} \right)
\end{eqnarray*}
so that in the limit $\eta\rightarrow 0^+$, for fixed $\nu<0$, $0<\beta<1$ and $a>0$, 
\begin{equation}
T(\nu,\eta)> \frac{\beta}{\nu}\ln \eta +  K_2(\beta) + O(\eta^2)
\label{eq:nultzerolower}
\end{equation}
where
\[
K_2(\beta)=\frac{\beta}{\nu}\left(-\frac{1}{4}-\ln|\nu|+\ln(-\ln(1-\beta))+\frac{1}{2}\ln\left(  1-\frac{\nu}{a^2}\right) \right) 
\]
\hfill$\square$
\end{proof}

\subsection{Scaling at bifurcation}

For the case $\nu=0$ where there is a bifurcation of the equilibrium at $x=0$ we obtain quite a different scaling. More precisely, 
\begin{lemma}
Suppose $\nu=0$ and pick any $0<\beta<1$. Then in the limit $\eta\rightarrow 0^+$, 
\begin{equation}
\frac{\sqrt{\pi}}{2}<\frac{T(\nu,\eta)}{\frac{1}{\eta}} <\sqrt{\frac{\pi}{2}}.
\label{eq:nultzero}
\end{equation}
\end{lemma}

\begin{proof}
Set $\nu=0$ (so that $\alpha=0$), then the estimate (\ref{eq:Tbounds}) gives an upper bound
\begin{eqnarray*}
T(\nu,\eta) &<& \frac{1}{\eta} \int_{z=0}^{\frac{2a^2}{\eta}} \frac{1-\exp(-2z^2)}{2z^2} \,dz\\
& = & \frac{1}{\eta} \left[\frac{1-\exp(-2z^2)}{-2z} \right]_0^{\frac{2a^2}{\eta}}+\frac{2}{\eta}\int_0^{\frac{2a^2}{\eta}}\exp(-2z^2) \, dz \\
& = & -\frac{1}{4a^2}+\frac{\exp(-8a^4/\eta^2)}{4a^2}+\frac{1}{\eta}\sqrt{\frac{\pi}{2}}\left(1-\mathrm{erfc}\left(\frac{2\sqrt{2}a^2}{\eta} \right) \right) 
\end{eqnarray*}
where $\mathrm{erfc}$ is the complementary error function. Using the asymptotic expansion for $\mathrm{erfc}$ for large $X$ given by
\[
\mathrm{erfc}(X)=\frac{\exp(-X^2)}{X\sqrt{\pi}}\left(1-\frac{1}{2X^2} + \dots \right)
\]
we find the lowest order terms for $T(\nu,\eta)$ are
\begin{eqnarray*}
T(\nu,\eta) &<& \frac{1}{\eta} \sqrt{\frac{\pi}{2}}-\frac{1}{4a^2}+\frac{\eta^2\exp(-8 a^{4}/\eta^{2})}{64a^6}+O\left(\eta^4\exp(-8a^{4}/\eta^{2})\right)
\end{eqnarray*}
as $\eta \rightarrow 0^+$. 
A similar computation for the lower bound gives
\begin{eqnarray*}
T(\nu,\eta) &>& \frac{1}{2\eta} \int_{y=0}^{\frac{a^2}{\eta}} \frac{1-\exp(-y^2)}{y^2} \,dy\\
& = & \frac{1}{2\eta} \left[\frac{1-\exp(-y^2)}{-y} \right]_0^{\frac{a^2}{\eta}}+\frac{1}{\eta}\int_0^{\frac{a^2}{\eta}}\exp(-y^2) \, dy \\
& = & -\frac{1}{2a^2}+\frac{\exp(-a^4/\eta^2)}{2a^2}+\frac{1}{\eta}{\frac{\sqrt{\pi}}{2}}\left(1-\mathrm{erfc}\left(\frac{a^2}{\eta} \right) \right) \\
& = & \frac{1}{\eta}{\frac{\sqrt{\pi}}{2}}-\frac{1}{2a^2}+\frac{\eta^2\exp(-a^{4}/\eta^{2})}{4a^6}+O(\eta^4 \exp(-a^4/\eta^2)).
\end{eqnarray*} 
\hfill$\square$
\end{proof}

The estimate (\ref{eq:Tbounds}) also means that we have a particularly tractable scaling if we look at the limit on fixing $\alpha$ (so that $\nu=\alpha \eta$) and taking $\eta\rightarrow 0^+$:
\begin{equation}
\label{eq:upperbdalpha}
T(\nu,\eta) < \frac{C(\alpha)}{\eta}+O(1)
\end{equation}
where
$$
C(\alpha)=\int_{z=0}^{\infty} \frac{1-\exp(2z(\alpha-z))}{2z(z-\alpha)} \,dz.
$$
is a constant that is small for $\alpha<0$ and grows very quickly for $\alpha>0$.
More generally this suggests that
\begin{equation}
\label{eq:scalingalpha}
T(\nu,\eta) \approx \frac{C(\alpha)}{\eta}+O(1)
\end{equation}
for some $C(\alpha)>0$ with $C(\alpha)\rightarrow 0$ as $\alpha\rightarrow -\infty$ and $C(\alpha)\rightarrow \infty$ as $\alpha\rightarrow \infty$. We believe the upper bounds is closer than the lower bounds, i.e. numerical evidence (see Figure~\ref{fig:escapes}) suggests that
$$
C(0)\leq \sqrt{\frac{\pi}{2}}=1.253314.
$$

\subsection{Scaling for excitable connections}

For $0<\nu<2a^2$ and $\eta>0$ (so that $\alpha>0$), if $\eta\rightarrow 0^+$ then we are in the standard Kramers case. We can compute this directly from~\eqref{eq:escapes}, that is
\begin{align*}
T(\nu,\eta) =  & \frac{2}{\eta^2} \int_{x=0}^{a} \int_{y=0}^{x} \exp \frac{\nu(x^2-y^2)+(y^4-x^4)}{\eta^2} \,dy\,dx \\
 = & \frac{2}{\eta^2} \int_{x=0}^{a} \exp \frac{\nu x^2-x^4}{\eta^2} \int_{y=0}^{x} \exp \frac{- \nu y^2+ y^4}{\eta^2} \,dy\,dx
\end{align*}
We note that the integrand of the first integral is maximal at $0<\sqrt{\nu/2}<a$, and that of the second at $0$. 
We approximate the significant contribution to the second integral over the range $0<y<\sqrt{\nu}$ and write $\mathrm{erf}(x)=\frac{2}{\sqrt{\pi}} \int_{s=0}^{x} \exp(-s^2)\,ds$ so that
\begin{align}
T(\nu,\eta) \approx & \frac{2}{\eta^2} \int_{x=-\infty}^{\infty} \exp \left( \frac{\nu^2}{4\eta^2} -\frac{2\nu}{\eta^2}\left(x-\sqrt{\frac{\nu}{2}}\right)^2\right) \, dx \int_{y=0}^{\sqrt{\nu}} \exp \frac{-\nu y^2}{\eta^2}  \,dy \nonumber \\
\approx& \frac{2}{\eta^2}  \exp \left( \frac{\nu^2}{4\eta^2} \right) \sqrt{\frac{\pi \eta^2}{2\nu}} \frac{\sqrt{\pi}}{2} \sqrt{\frac{\eta^2}{\nu}}\mathrm{erf}\left(\frac{\nu}{\eta}\right) \nonumber \\
\approx & \frac{\pi}{\nu\sqrt{2}}  \exp \left( \frac{\nu^2}{4\eta^2} \right) \label{eq:Kramersnuplus}
\end{align}
for fixed $\nu>0$ and $\eta\rightarrow 0^+$, which corresponds to the formula (\ref{eq:Kramers}).

In fact, an approximation that is valid over a larger range of $\eta$ can be found as follows, using an explicit lower bound. We write $(x,y)=r(\cos \theta,\sin \theta)$, and $s=r^2$. Then, we assume that $\nu<2a^2$, use the fact that $\exp(a)\leq \exp(b)\leq 1$ if $a\leq b\leq 0$, and the inequalities
$$
-2\theta^2\leq \cos 2\theta-1,\quad \cos 2\theta\leq 1
$$
on $\theta\in[0,\pi/4]$ to show that 
\begin{align}
T(\nu,\eta) = & \frac{2}{\eta^2} \int_{x=0}^{a}\int_{y=0}^{x} \exp \left( \frac{\nu(x^2-y^2)+(y^4-x^4)}{\eta^2} \right)  \,dy\,dx \nonumber \\
> & \frac{1}{\eta^2}\int_{s=0}^{a^2}\int_{\theta=0}^{\pi/4} \exp \left( \frac{s(\nu-s)\cos 2\theta}{\eta^2}\right) \,\,ds\,d\theta\nonumber
\\
= & \frac{1}{\eta^2} \int_{s=0}^{a^2}\int_{\theta=0}^{\pi/4} \exp \left( \frac{(\nu^2/4-(s-\nu/2)^2)\cos 2\theta}{\eta^2}\right) \,\,ds\,d\theta\nonumber
\\
= & \frac{1}{\eta^2}\exp\left(\frac{\nu^2}{4\eta^2}\right) \int_{s=0}^{a^2}\int_{\theta=0}^{\pi/4} \exp \left( \frac{\nu^2}{4\eta^2}(\cos 2\theta -1)\right)  \exp\left(-\frac{(s-\nu/2)^2\cos 2\theta}{\eta^2}\right) \,\,ds\,d\theta\nonumber
\\
> & \frac{1}{\eta^2}\exp\left(\frac{\nu^2}{4\eta^2}\right) \int_{\theta=0}^{\pi/4} \exp \left(- \frac{\nu^2\theta^2}{2\eta^2}\right) \,d\theta \int_{s=0}^{a^2}\exp\left(-\frac{(s-\nu/2)^2}{\eta^2}\right) \,\,ds.\nonumber
\end{align}
Evaluating these integrals we have
\begin{align}
T(\nu,\eta) > & 
 \frac{\pi\sqrt{2} }{4\nu} \exp \left( \frac{\nu^2}{4\eta^2}\right) \mathrm{erf}\left(\frac{\pi\nu\sqrt{2}}{8\eta}\right)\left[\mathrm{erf}\left(\frac{\nu}{2\eta}\right)+ \mathrm{erf}\left(\frac{2a^2-\nu}{2\eta}\right)\right].
\end{align}
Hence, for fixed $\nu>0$ and $2a^2>\nu$ we have $\mathrm{erf}(\nu/(2\eta))\approx\mathrm{erf}((2a^2-\nu)/(2\eta))\approx 1$ in the limit $\eta\rightarrow 0^+$, and hence
\begin{equation}
T(\nu,\eta) \geq \frac{\pi}{\nu\sqrt{2}}  \exp \left( \frac{\nu^2}{4\eta^2} \right).\label{eq:Kramersnu}
\end{equation}
i.e. Kramer's formula (\ref{eq:Kramersnuplus}) is a lower bound in this case. On the other hand, if both $\nu$ and $\eta$ are small, and $\nu/\eta$ is $O(1)$  then $\mathrm{erf}\left(\frac{\pi\nu\sqrt{2}}{8\eta}\right)\approx \frac{\nu\sqrt{\pi}}{4\eta \sqrt{2}}$ and so
\begin{align}
T(\nu,\eta) \geq & \frac{\pi}{\nu\sqrt{2}}  \exp \left( \frac{\nu^2}{4\eta^2} \right)\mathrm{erf}\left(\frac{\pi\nu\sqrt{2}}{8\eta}\right)\approx
\frac{\pi^{3/2}}{8\eta}  \exp \left( \frac{\nu^2}{4\eta^2} \right)
\label{eq:Kramersnuplustwo}
\end{align}
In summary, for small but fixed $\nu$ and $\eta\rightarrow 0^+$ we have 
$$
T(\nu,\eta) \approx \frac{K_1}{\nu}  \exp \left( \frac{\nu^2}{4\eta^2} \right)
$$
while for small $\nu$ and $\eta$ but $\nu/\eta$ being $O(1)$ we have 
\begin{equation}
T(\nu,\eta) \approx \frac{K_2}{\eta
}  \exp \left( \frac{\nu^2}{4\eta^2} \right)
\label{eq:Kramersbigeta}
\end{equation}

In Figure~\ref{fig:summaryres} we summarise the scalings we have obtained for mean residence time in the low noise limit, near bifurcation from heteroclinic to excitable connections, while in Figure~\ref{fig:escapes} we numerically verify examples of these scalings.

\begin{figure}%
\centerline{
\setlength{\unitlength}{5cm}
\begin{picture}(2,1)
\put(0,0){\includegraphics[height=5cm]{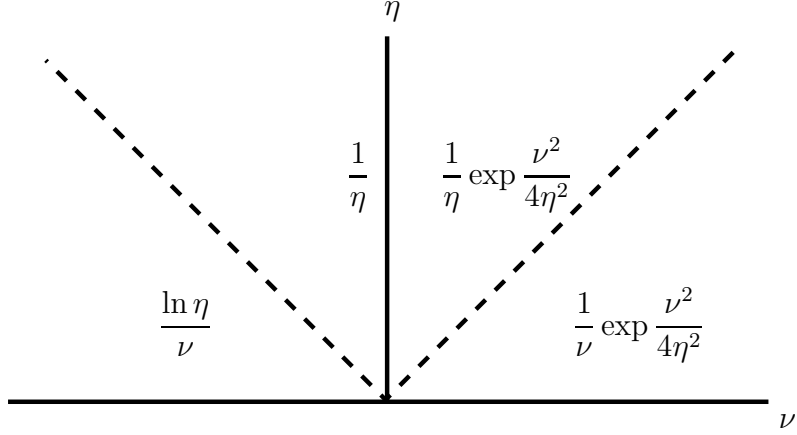}}
\put(2.05,-0.05){$\nu$}
\put(1,1.05){$\eta$}
\put(0.4,0.2){$\dfrac{\ln\eta}{\nu}$}
\put(0.9,0.6){$\dfrac{1}{\eta}$}
\put(1.5,0.2){$\dfrac{1}{\nu}\exp{\dfrac{\nu^2}{4\eta^2}}$}
\put(1.15,0.6){$\dfrac{1}{\eta}\exp{\dfrac{\nu^2}{4\eta^2}}$}
\end{picture}
}
\caption{Schematic showing the asymptotic scalings of residence time $T(\nu,\eta)$ at a node, considered in the plane for fixed $\nu$ the leading eigenvalue at the node and $\eta$ the noise strength going to zero. In all cases $T$ is finite for $\eta>0$ but $T(\nu,\eta)\rightarrow \infty$ as $\eta\rightarrow 0^+$ for fixed $\nu$; how fast this diverges depends qualitatively on whether the associated connection is heteroclinic ($\nu<0$) or excitable ($\nu>0)$.}
\label{fig:summaryres}%
\end{figure}

\subsection{Simulation of escape for a one dimensional SDE}
\label{sec:onedimsims}

To illustrate the above scalings, we consider the SDE (\ref{eq:sdeescape}) for the potential \eqref{eq:mpotential}, i.e.
\begin{equation}
dy = (-y^4+2y^2-\nu)y dt+\eta dw.
\label{eq:yonly}
\end{equation}
We choose an $a>0$ that is away from all equilibria (typically we use $a=0.5$) and numerically compute the mean escape time
\begin{equation}
T(\nu,\eta) =  \langle \{T~:~|y(T)|=a \mbox{ and }|y(t)|<a \mbox{ for all }0<t<T\}\rangle
\end{equation}
where the mean is taken over the distribution of initial $y(0)$ and over realizations of the noise process in (\ref{eq:yonly}). Using a stochastic Euler approximation with timestep $h=0.01$ and $n=1000$ realizations for each calculation gives approximations of $T(\nu,\eta)$ as a function of $\nu$ and $\eta$; see Figure~\ref{fig:escapes}. In the three cases we verify agreement of the measured mean residence times with the predicted scalings in three cases. For $\nu=-0.01$ we show the best fit (black curve) to $T=A \ln(\eta)+B$ with $A=-96$ and $B=-369$; this compares well with the prediction $A=1/\nu=-100$ and $B=(-1-\ln(|\nu|)+\ln(1-\ln/(2a^2)))/\nu)=-362.5$ from equation \eqref{eq:nultzeroupper}.
For $\nu=0$ we show the best fit (red curve) to $T=A/\eta+B$ with $A=1.152$ and $B=2.378$, again, this compares well with the prediction $A=\sqrt{\pi/2}=1.2533$ from equation \eqref{eq:scalingalpha}.
For $\nu=0.01$ we show the best fit (blue curve) to $T=A/\eta\exp(B/\eta^2)$ with $A=1.4424$ and $B=2.027\times 10^{-5}$ - compare with $B=\nu^2/4=2.5\times 10^{-5}$ in \eqref{eq:Kramersbigeta}. 

In the first two cases we also find good agreement between fitting parameters and predicted values. In the third case we do not have a tight asymptotic fit but nevertheless, empirically there is a good fit to the scaling formula over this range. For the third case, we expect that the Kramers formula is more accurate for the range $2\eta<\nu$, though the timescales become extremely long.

\begin{figure}%
\begin{center}
\includegraphics[width=12cm]{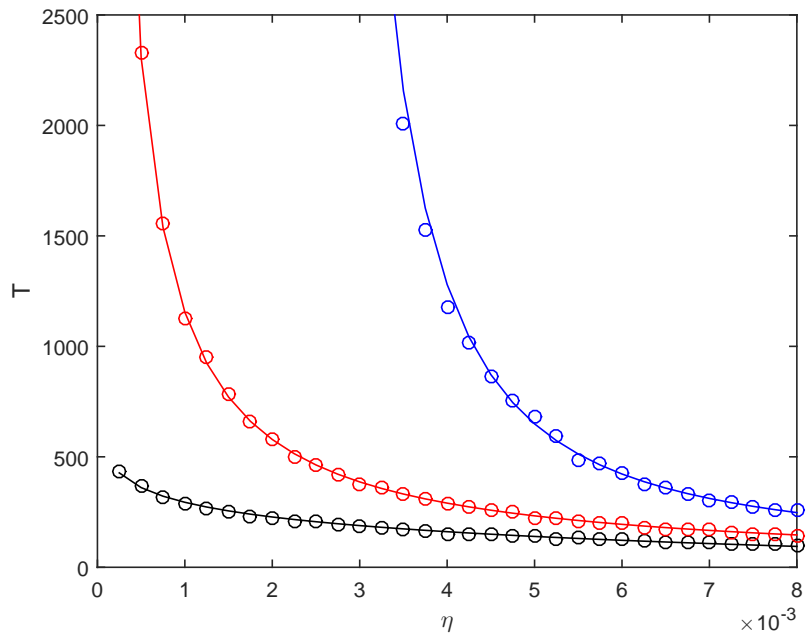}
\end{center}
\caption{The open circles show numerical estimation of the mean residence time $T$ from $y=0$ in (\ref{eq:yonly}) as a function of noise $\eta$ and parameter $\nu$: black is $\nu=-0.01$, red is $\nu=0$ and blue is $\nu=0.01$. We use first arrival at $|y|=a=0.5$ to detect escape.
For $\nu=-0.01$ we show the best fit (black curve) to $T=A \ln(\eta)+B$.
For $\nu=0$ we show the best fit (red curve) to $T=A/\eta+B$.
For $\nu=0.01$ we show the best fit (blue curve) to $T=A/\eta\exp(B/\eta^2)$: see text for more details.}%
\label{fig:escapes}%
\end{figure}

\section{Transition probabilities and multiple independent escape processes}
\label{sec:switching}

In order to understand transition probabilities we must consider noisy network attractors where there is more than one possible connection from a given node. We start by discussing the bi-directional ring from Section~\ref{sec:3graph}. For small noise, it turns out that the switching can be well-approximated by multiple independent escape processes: see Section~\ref{sec:multiplesc}. This gives evidence supporting Conjecture~\ref{conj:main}.

\subsection{Example: bi-directional ring around three nodes} \label{sec:example}

Consider a noisy network attractor that realises the bi-directional ring shown in Figure~\ref{fig:3graphs}(b). For simplicity we assume there is full permutation symmetry of the three nodes in the noise-free system. This means that there are two independent outgoing connections at each node, and the dynamics on these connections is the same. However, we choose the noise amplitude amplitudes $\eta_{cw}>0$ for clockwise (resp. $\eta_{acw}>0$ for anticlockwise) transitions that may be different.

Using the system detailed in Appendix~\ref{app:bidirthree}, equation~\eqref{eq:C3system2} we perform some numerical simulations on varying $\eta_{cw}$, $\eta_{acw}$ and a parameter $\nu$ for the case where connections in the network are (a) heteroclinic $\nu<0$ (b) at bifurcation $\nu=0$ and (c) excitable $\nu>0$. In each case, we verify that the residence time and transition probabilities for clockwise (resp. anticlockwise) transitions appear to vary continuously and monotonically with the amplitudes $\eta_{cw}$ (resp. $\eta_{acw}$): Figure~\ref{fig:example3bidir} shows this for the bifurcation case (b). Figure~\ref{fig:3plots} (a), (c) and (e) show all three cases.

\begin{figure}
	\centerline{
{\includegraphics[width=80mm]{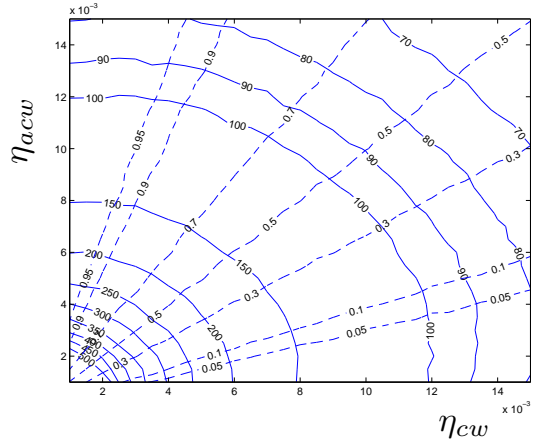}}
	}
	\caption{Contours showing the mean residence time (solid lines) and transition probability $\pi_{acw}$ (dashed lines) for anticlockwise transitions for the system described in Section~\ref{sec:example} for the critical (bifurcation) case ($\nu=0$), as a function of the noise amplitudes $\eta_{cw}$ (resp. $\eta_{acw}$) that excite transitions in the clockwise (resp. anticlockwise) directions. 
	Note that the mean residence times are constant on closed curves around the origin while the lines of constant transition probability are approximately radial. Similar plots are obtained for parameters that give heteroclinic or excitable networks in Figure~\ref{fig:3plots}.}
	\label{fig:example3bidir}
\end{figure}

We find that the choice of connection is well modelled by  multiple independent escape processes. This gives insight into the more general problem. We note that for a fully nonlinear SDE with a noisy network attractor it may however be very difficult to estimate transition probabilities analytically.

\subsection{Multiple independent escape processes}
\label{sec:multiplesc}

Consider $n$ independent escape processes, where we escape in direction $m=1,\ldots,n$ after a time given by a continuous random variable $T_m>0$ with distribution $\rho_m(T_m)$. A {\em multiple independent escape process} means that the first ``escape'' stops the process and identifies one particular direction of escape. More precisely, we say there is {\em  escape in direction $k$ at time $T$} in the case $T=T_k<T_i$ for all $i\neq k$. The distributions of random variables giving the first escape time $T$ and the escape direction $k$ are:
\begin{equation}
T=\min(T_1,\ldots,T_n),~~k=\argmin(T_1,\ldots,T_n).
\label{eq:firstescape}
\end{equation}
Let $\rho$ be the distribution of the random variable $T$ and $\kappa(k)$ the probabilities of the discrete random variable $m=k$.  The random variables $T$ are called order statistics \cite{David03}, and one can find these from the distribution $\rho_m$ of the individual escapes $T_m$ as follows:

\begin{lemma}
\label{lem:firstescape}
The distribution $\rho(T)$ of first escape times and the probability $\kappa(k)$ are given by
\begin{eqnarray*}
\rho(T) &=& \sum_{k=1}^n \left[\rho_k(T) \prod_{m=1,m\neq k}^{n} \int_{t_m=0}^{t_k} (1-\rho_m(t_m))\,dt_m \right]\\
\kappa(k) &= &\int_{t=0}^{\infty} \rho_k(t_k)\prod_{m=1,m\neq k}^{n} \int_{t_m=0}^{t_k} (1-\rho_m(t_m))\,dt_m \,dt_k
\end{eqnarray*}
\end{lemma}

\proof This can be seen by noting that the distribution of $T$ is the sum distributions for the probability that the first escape happens in the $k$th direction at time $T$.
\qed

~

Standard results on order statistics imply that if the $T_k$ are all exponentially distributed as
$$
\rho_k(T_k)=\frac{1}{r_k}\exp(-T_k/r_k)
$$
for $r_k>0$ then
\begin{equation}
\rho(T) = \frac{1}{r}\exp(-T/r)
\label{eq:exptau}
\end{equation}
where $1/r=\sum_{m=1}^{n} (1/r_m)$. In other words, if the $T_m$ are exponentially distributed then so is $T$, and the mean rate of escape that is the average of the rates of escape of the individual processes. For this case we can compute
\begin{equation}
\kappa(k) = \frac{1/{r_k}}{\sum_{m=1}^{n} 1/{r_m}}= \frac{r}{r_k}.
\label{eq:expkappa}
\end{equation}
In this case the process with the fastest mean escape time will be the direction where escapes are most frequent.

For more general distributions for the individual processes, even if they remain independent, $\rho$ and $\kappa$ are not usually explicitly computable from the integral forms and indeed may be counter-intuitive for some sets of distribution of $\rho_k$, especially if they are multi-modal or the tails are of different weight. For example, suppose $n=2$ with $\rho_1(t_1) = 0.9 \delta(t_1-1)+0.1\delta(t_1-100)$ and $\rho_2(t_2)=\delta(t_2-2)$. Then $E(t_1)=10.9$ and $E(t_2)=2$, so mean escape time in direction 1 is much slower than in direction 2. On the other hand, the probability of the first escape occurring in direction 1 is higher than $0.9$!

For the noise-induced escape processes we consider, the distributions are determined by escapes from potential wells (and have exponential tails for low noise) or from near saddles (and may have faster-decaying tails); in both cases the distributions will not be exponential though for escape from the potential well it will have an exponential tail where the rate corresponds to the Kramers escape rate.

\subsection{Multiple escape times and transition probabilities}
\label{sec:switchingodes}

We illustrate a multiple independent escape process for a system of $n$ SDEs
\begin{equation}
dy_{k} = (-y_{k}^4+2y_{k}^2-\nu_{k})y_{k} dt+\eta_{k} dw_{k}.
\label{eq:yk}
\end{equation}
where $k=1,\ldots,n$ and $w_{k}$ are independent Brownian processes and $\nu_k,\eta_k$ are parameters and assume that we start at some $y(0)=0$; cf (\ref{eq:sdeescape}). We choose a $K>0$ that is away from all equilibria (typically we use $K=0.5$). There is escape in direction $k$ at time $\tau_k$ if 
$$
|y_k(\tau_k)|=K,\mbox{ and } |y_{k}(t)|<K\mbox{ for all }0<t<\tau_k
$$
and define $\tau$ to be the first escape and $k$ the direction of first escape as in (\ref{eq:firstescape}).

Using the multiple independent escape process (\ref{eq:yk}) with $n=2$ we can approximate the behaviour of switching for the example of the bi-directional ring on three nodes discussed in Section~\ref{sec:example}.

Figure~\ref{fig:3plots}, left column (a,c,e) shows the mean residence times and switching probabilities for the noisy network attractor on varying $\eta_1=\eta_{cw}$ and $\eta_2=\eta_{acw}$. The right column (b,d,e) shows the mean first escape times and escape probabilities for the multiple escape process (\ref{eq:yk}) with the corresponding $\eta$. The computations are performed using a stochastic Euler integrator with timestep $h=0.05$. The values of $\eta_k$ are discretized into $29$ steps in each direction.

Contours of mean first escape time $T$ (solid lines) and equiprobability $\pi_{cw}$ (dashed lines) are shown in Figure~\ref{fig:3plots}. Subfigures (a,b) with $\nu<0$ corresponds to $y=0$ being linearly unstable and a noisy heteroclinic connection with two outgoing directions, subfigures (e,f) the case $\nu>0$ corresponds to $y=0$ being a sink with a small basin and a noisy excitable connection with two outgoing directions, and subfigures (c,d) are the bifurcation case $\nu=0$. Observe that there is good quantitative and qualitative agreement in all three cases illustrated in Figure~\ref{fig:3plots}.  From these figures, we note the following:
\begin{itemize}
\item
The mean residence time $T$ decreases monotonically with $\eta_{aw}$ or $\eta_{acw}$. Also, $T\rightarrow \infty$ as $\max(\eta_{cw},\eta_{acw})\rightarrow 0^+$ in all cases.
\item
The transition probability $\pi_{cw}$ increases monotonically with $\eta_{cw}$ for fixed $\eta_{acw}$, moreover $\pi_{cw}\rightarrow 0$ as $\eta_{cw}/\eta_{acw}\rightarrow 0$ and $\eta_{cw}\rightarrow 1$ as $\eta_{cw}/\eta_{acw}\rightarrow 1$.
\end{itemize}
In summary, the numerical results in the left column of Figure~\ref{fig:3plots} suggest that Conjecture~\ref{conj:main} holds for all three cases of this symmetrised network where only the noise amplitudes break the symmetry. More precisely, by suitable choice of noise amplitudes $\eta_{cw},\eta_{acw}$ one can realise any transition probability $\pi_{cw}\in(0,1)$ and any sufficiently long mean residence time $T>0$.  As expected from the discussion in Section~\ref{sec:residence} the scaling properties of $\pi_{cw}$ and $T$ depend strongly on whether the network is heteroclinic or excitable, near the boundaries $\pi_{cw}=0,1$ and $T=\infty$.

\begin{figure}%
	\begin{center}
\includegraphics[width=130mm]{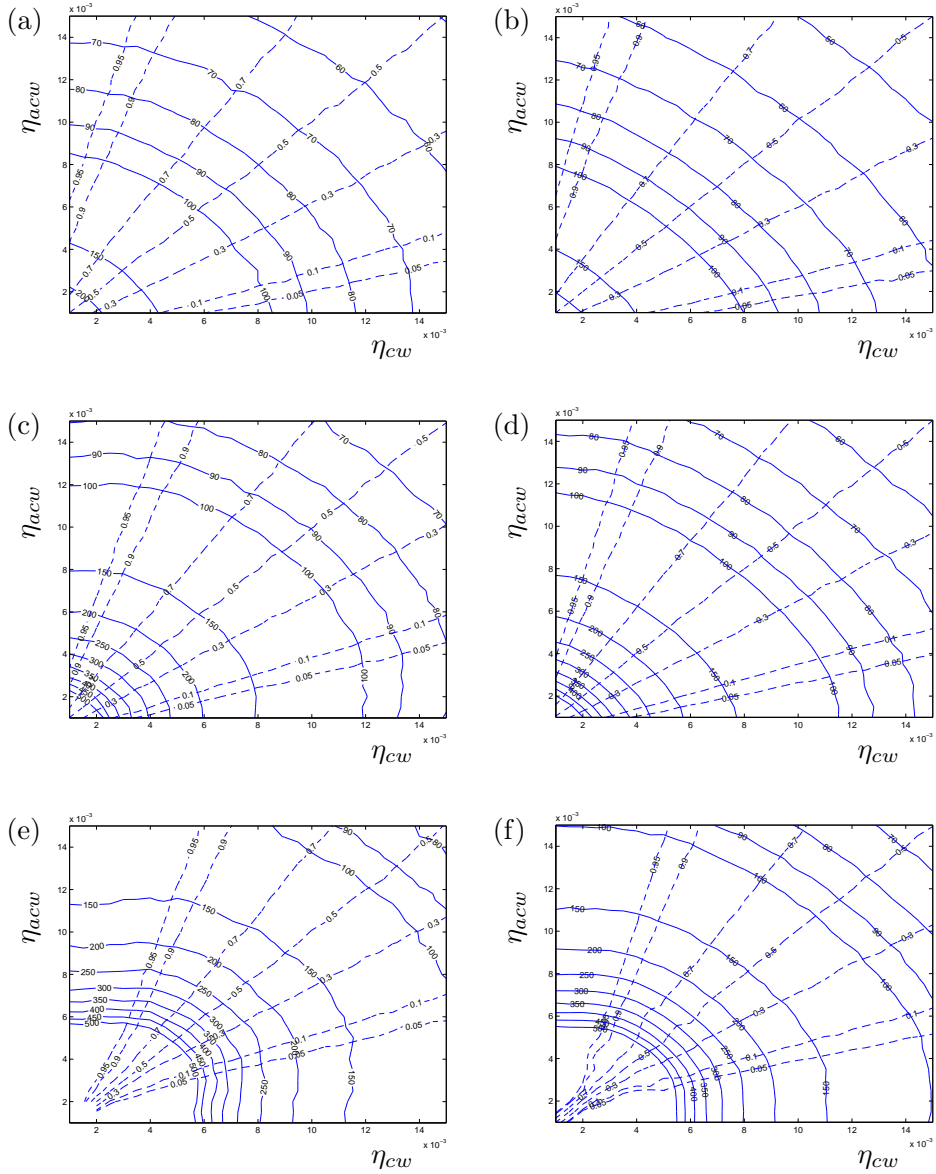}
	\end{center}
	\caption{(a,c,e): Contours of mean residence time $T$ (solid lines) and transition probability $\pi_{cw}$ (dashed lines) for clockwise motion in the system~\eqref{eq:C3system2} describing the bi-directional ring shown in Figure~\ref{fig:3graphs}(b). (b,d,f): Contours of mean first escape time (solid lines) and probability of first escape in direction $2$ (dashed lines) for the system (\ref{eq:yonly}) on varying $\eta_1=\eta_{cw}$, $\eta_2=\eta_{acw}$ with (a,b) $\nu=-0.01$, (c,d) $\nu=0$ and (e,f) $\nu=0.01$. Note in all cases there is good agreement. For the excitable case (e,f) there is a steep rise in residence time in the region where $\max\eta_i<\nu/2=0.005$}%
	\label{fig:3plots}%
\end{figure}

\subsection{Approximating multiple escape processes for excitable networks}

The formulae from Lemma~\ref{lem:firstescape} suggest that in general one cannot obtain the mean escape time or direction of escape from a multiple escape process simply from knowledge of the mean escape time of each process: one needs knowledge of the distribution of escape times for the individual processes. However, in the case of an excitable network where there are approximately exponential distributions of residence times, this is possible. 

Figure~\ref{fig:expplots}(a) shows the mean residence times and transition probabilities for the excitable case Figure~\ref{fig:3plots}(f) and Figure~\ref{fig:expplots}(b) shows that for the distribution (\ref{eq:exptau},\ref{eq:expkappa}), using a best fit to the exponential tail of a single escape process. More precisely, we use $T\approx (A/\eta) \exp(B/\eta^2)$ as in (\ref{eq:Kramersbigeta}) for the escape time $T$ in one direction with $\eta=0.01$, using $A=1.4$ and $B=2.430\times 10^{-5}$, cf. Figure~\ref{fig:escapes}.

\begin{figure}%
\begin{center}
\includegraphics[width=130mm]{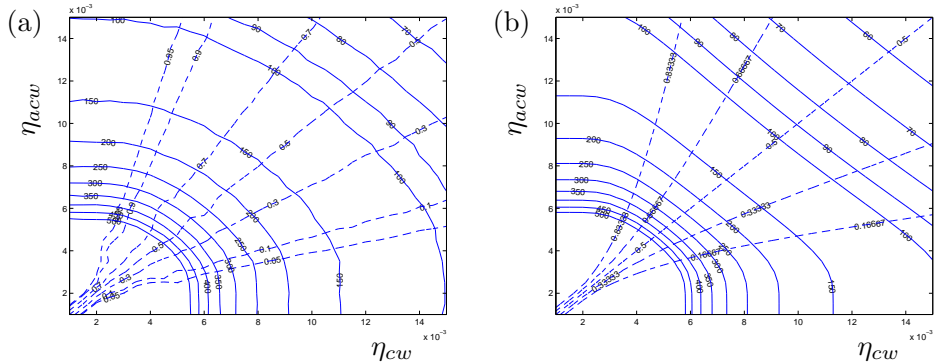}
\end{center}
\caption{(a) Contours showing mean first escape time $T$ (solid lines) and probability of first escape in direction $2$ (dashed lines) for the system (\ref{eq:yonly}) on varying $\eta_1=\eta_{cw}$, $\eta_2=\eta_{acw}$ with $\nu=0.01$, as in Figure~\ref{fig:3plots}(f). (b) Contours of probabilities of escape in direction $2$ (dashed lines) and mean first escape time (solid lines) for exponential distributions fitted to empirically determined means: see text for details.}
\label{fig:expplots}%
\end{figure}

\section{Discussion}
\label{sec:discuss}

There are a number of subtle effects of noise on heteroclinic networks that have been discussed in previous work~\cite{Bakhtin,ArmStoKir03}. This paper expands and extends this to noisy excitable networks that are created by bifurcation from heteroclinic in the noise-free case. Clearly, the mean properties of the macroscopic randomness (that is, the residence times at nodes and the transition probabilities between nodes) depend on the (anisotropic) noise amplitudes. Conjecture~\ref{conj:main} suggests that, {\em vice versa}, one can select noise amplitudes to approximate a Markov process on the network as a noisy network with given transition probabilities and mean residence times. We verify this for a simple case of a bi-directional ring network; in future work we will to explore this for more complex networks in the presence of noise perturbations.

It was noted in \cite{AshPos15} that noisy heteroclinic cycles will have approximately log normal distributions of residence times \cite{StoHol90} while excitable cycles will have exponential tails to the distributions of residence times. We believe that the \emph{distributions} of macroscopic fluctuations are much more difficult to determine from the microscopic noise distributions than the means. Quite complex distributions may result, for example if there are multiple connections between the same pair of nodes. 

For the weak noise case, we find numerical evidence that the residence times and transition probabilities can be characterised by modelling the transitions between nodes in the network as multiple independent escapes processes. This is not too much of a surprise, at least if the Jacobian is diagonalisable and the principal axes of noise correspond to these directions - an analysis in the general case will probably be much more complicated.  In Section~\ref{sec:residence} we use the approximation of a single escape process to obtain some asymptotic scalings of mean residence times - it will be a challenge to find more accurate and justified asymptotic expressions, especially for the excitable case, and to obtain asymptotic expressions for higher moments in the distributions.

We work here with networks where the transition probabilities are memoryless - that is, they are well-modelled by a first order Markov chain. In previous work~\cite{AshPos13} we discussed an example where this is not the case and noise-induced `lift-off'~\cite{ArmStoKir03} causes longer-term correlations in the sequence of nodes visited. It will be a challenge to understand properties of the long-term correlations and, for instance, whether they affect the scalings of residence times at the nodes.

There are many more open problems that deserve a detailed analysis - indeed, an appropriate definition of a noisy network attractor is still debatable. Should this be a statistical attractor whose empirical measures are close to delta functions on the nodes, or is a more stringent definition appropriate? Given a good definition, progress on Conjecture~1 may be possible in a general setting. 

Finally, we mention some potential applications. Heteroclinic network models have been used for modelling cognitive functions \cite{ashwin_karabacak_nowotny_2011,bick_rabinovich_2010,komarov_osipov_suykens_09,neves_timme_12,Rabinovich2001} as they have the ability to perform finite state computations, as well as the capacity to translate microscopic random fluctuations into macroscopic randomness. This randomness is manifested both in terms of the residence times at nodes of the network and in terms of the  transition probabilities between nodes and hence choice of possible paths around the network. In this paper, we have highlighted that this work should extend in a natural way to excitable networks. Excitable networks may indeed be a more natural way to understand computations.

\subsection*{Acknowledgments}
We thank many people for stimulating conversations that contributed to the development of this paper: in particular Chris Bick, Nils Berglund, Mike Field, John Terry, Ilze Ziedins. We thank the London Mathematical Society for support of a visit of CMP to Exeter, and the University of Auckland Research Council for supporting a visit of PA to Auckland during the development of this research. PA gratefully acknowledges the financial support of the EPSRC via grant EP/N014391/1.

\newpage

\appendix

\section{Construction of dynamics realising network attractors}
\label{app:realise}
\label{app:bidirthree}
\label{app:3graphs}

As outlined in \cite{AshPos15} we consider a system of coupled ODEs that realises a arbitrary directed graph $G=(\V,\E)$ as a heteroclinic or an excitable network (depending on parameters): 
\begin{equation}
\label{eq:realiseode}
\begin{split}
\d p_{j} & =  [p_j(F(1- p^2)+ D(p_j^2p^2-p^4))+ E(-Z^{(o)}_j(p,y)+Z^{(i)}_j(p,y)]\d t+ \eta_p \d w_{y,j}\\
\d y_{k} & = [G\left(y_{k},A-B p_{\alpha(k)}^2 +C (y^2-y_{k}^2\right))]\d t+ \eta_{y,k} \d w_{p,k}%
\end{split}
\end{equation}
for $j=1,\cdots,n_v$ and $k=1,\cdots,n_e$, where $p^2=\sum_{j=1}^{n_v} p_j^2$, $p^4=\sum_{j=1}^{n_v} p_j^4$, $y^2=\sum_{j=1}^{n_e} y_j^2$ and $A,B,C,D,E,F$ are constants. The function $G$ is defined by
\begin{equation}
G(y_{k},\lambda)= -y_{k}\left((y_{k}^2-1)^2+\lambda\right)
\label{eq:G}
\end{equation}
while the inputs to the $p_j$ cells from the $y$ cells are:
\begin{equation}
\begin{split}
Z^{(o)}_j(p,y) &= \sum_{\{k~:~\alpha(k)=j\}} -y_k^2p_{\omega(k)}p_j\\
Z^{(i)}_j(p,y) &= \sum_{\{k'~:~\omega(k')=j\}} y_{k'}^2p_{\alpha(k')}^2.
\end{split}
\end{equation}
For $\eta\equiv 0$ the system is an ODE and $\xi_{j}$ denote the unit basis vectors $(p,y)\in \R^{n_v+n_e}$: the first $n_v$ correspond to unit vectors where one of the $p_j$ is non-zero.  As shown in \cite{AshPos15}, the subspaces
$$
P_{\ell}=\{(p,y)~:~y_k=0~\mbox{ if }k\neq \ell~\mbox{ and }p_j=0~\mbox{ if }j\neq\alpha(\ell)\mbox{ or }\omega(\ell)\}
$$
for $\ell=1,\ldots,n_e$ are invariant for the flow generated by system (\ref{eq:realiseode}) and for suitable choice of parameters contain connections that realise the graph $G$ as a heteroclinic/excitable network embedded in phase space. For $\eta_{y,k}> 0$ there will be noise-induce motion around the network.

We choose a fixed noise amplitude $\eta_p=10^{-3}$ for the $p$ variables and default parameters
\begin{equation}
A=0.5,~B=1.5-\nu,~C=2,~ D=10,~~E=4,~F=2.
\label{eq:systemparams}
\end{equation}
For $\nu<0$ close to zero this realises a heteroclinic network, while for $\nu>0$ close to zero it realises an excitable network with a small threshold. The case $\nu=0$ corresponds to bifurcation between the two types of network: see \cite[Fig. 4]{AshPos15} for more details and justification that the networks are heteroclinic/excitable for these parameter values.

\subsection*{Unidirectional and bi-directional loops around three nodes.}

In order to realise noisy versions of heteroclinic or excitable networks for the graphs illustrated in Figure~\ref{fig:3graphs}(a,b) we consider the following systems of equations. For the uni-directional ring (a) we consider
\begin{equation}
\begin{split}
\dot{p}_1&=p_1(F(1-p^2)+D(p_1^2p^2-p^4))+E(-y_1^2 p_1p_2-y_6^2 p_1p_3 + y_2^2 p_3^2 )+\eta_p w_1\\
\dot{p}_2&=p_2(F(1-p^2)+D(p_2^2p^2-p^4))+E(-y_2^2 p_2p_3-y_4^2 p_2p_1 + y_1^2 p_1^2)+\eta_p w_2\\
\dot{p}_3&=p_3(F(1-p^2)+D(p_3^2p^2-p^4))+E(-y_3^2 p_3p_1-y_5^2 p_3p_2 + y_2^2 p_2^2)+\eta_p w_3\\
\dot{y}_1&=G(y_1,A-B p_1^2+C(y^2-y_1^2))+\eta_1 w_4\\
\dot{y}_2&=G(y_2,A-B p_2^2+C(y^2-y_2^2))+\eta_2 w_5\\
\dot{y}_3&=G(y_3,A-B p_3^2+C(y^2-y_3^2))+\eta_3 w_6\\
\end{split}
\label{eq:C3system1}
\end{equation}
while for the bi-directional ring (b) we consider
\begin{equation}
\begin{split}
\dot{p}_1&=p_1(F(1-p^2)+D(p_1^2p^2-p^4))+E(-y_1^2 p_1p_2-y_6^2 p_1p_3 + y_2^2 p_3^2 + y_4^2 p_2^2)+\eta_p w_1\\
\dot{p}_2&=p_2(F(1-p^2)+D(p_2^2p^2-p^4))+E(-y_2^2 p_2p_3-y_4^2 p_2p_1 + y_1^2 p_1^2 + y_5^2 p_3^2)+\eta_p w_2\\
\dot{p}_3&=p_3(F(1-p^2)+D(p_3^2p^2-p^4))+E(-y_3^2 p_3p_1-y_5^2 p_3p_2 + y_2^2 p_2^2 + y_6^2 p_1^2)+\eta_p w_3\\
\dot{y}_1&=G(y_1,A-B p_1^2+C(y^2-y_1^2))+\eta_1 w_4\\
\dot{y}_2&=G(y_2,A-B p_2^2+C(y^2-y_2^2))+\eta_2 w_5\\
\dot{y}_3&=G(y_3,A-B p_3^2+C(y^2-y_3^2))+\eta_3 w_6\\
\dot{y}_4&=G(y_4,A-B p_2^2+C(y^2-y_1^2))+\eta_4 w_7\\
\dot{y}_5&=G(y_5,A-B p_3^2+C(y^2-y_2^2))+\eta_5 w_8\\
\dot{y}_6&=G(y_6,A-B p_1^2+C(y^2-y_3^2))+\eta_6 w_9.
\end{split}
\label{eq:C3system2}
\end{equation}
We choose the standard set of parameters (\ref{eq:systemparams}) and vary both $\nu$ and (low amplitude) noise added to both $p_i$ and $y_i$ variables. We set $\nu=-0.01$ for the heteroclinic case and $\nu=0.01$ for the excitable case unless otherwise stated. Three the noise amplitudes are set the same
$$
\eta_1=\eta_2=\eta_3=\eta_{cw}.
$$
representing the amplitude of noise that promotes the clockwise transitions in Figure~\ref{fig:3graphs}(a,b). For the bi-directional case (b) we also set
$$
\eta_4=\eta_5=\eta_6=\eta_{acw}
$$
as the amplitude of the noise that promotes anticlockwise transitions. The one-dimensional observable
$$
S(t)= \sum_{k=1}^3 k y_k^2(t)
$$
has the property that $S(t)\approx k$ whenever the trajectory is near the equilibrium $\xi_k$ and can be used to observed the state of the system.

\newpage

\bibliographystyle{plain}

\end{document}